\documentclass[10pt, conference]{IEEEtran}

\makeatletter
\def\ps@headings{%
\def\@oddhead{\mbox{}\scriptsize\rightmark \hfil \thepage}%
\def\@evenhead{\scriptsize\thepage \hfil \leftmark\mbox{}}%
\def\@oddfoot{}
\def\@evenfoot{}}
\def\blfootnote{\xdef\@thefnmark{}\@footnotetext}
\makeatother

\usepackage{cite}
\usepackage{url}

\usepackage{graphicx}
\usepackage{algorithm}
\usepackage[noend]{algorithmic}
\usepackage{subfig}
\usepackage{amssymb, amsmath,graphicx,charter, latexsym}
\newtheorem{definition}{Definition}
\newtheorem{lemma}{Lemma}
\newtheorem{theorem}{Theorem}
\newtheorem{example}{Example}

\def\baselinestretch{0.86}
\begin{document}
\title{Scheduling Heterogeneous Real-Time Traffic over Fading Wireless Channels}
\author{
\IEEEauthorblockN{I-Hong Hou} \IEEEauthorblockA{CSL and Department
of CS\\University of Illinois\\Urbana, IL 61801,
USA\\ihou2@illinois.edu} \and \IEEEauthorblockN{P. R.
Kumar}\IEEEauthorblockA{CSL and Department of ECE\\University of
Illinois\\Urbana, IL 61801, USA\\prkumar@illinois.edu}
 }
\maketitle\blfootnote{This material is based upon work partially
supported by USARO under Contract Nos. W911NF-08-1-0238 and
W-911-NF-0710287, AFOSR under Contract FA9550-09-0121, and NSF under
Contract Nos. CNS-07-21992, ECCS-0701604, CNS-0626584, and
CNS-05-19535. Any opinions, findings, and conclusions or
recommendations expressed in this publication are those of the
authors and do not necessarily reflect the views of the above
agencies. }

\begin{abstract}
We develop a general approach for designing scheduling policies for
real-time traffic over wireless channels. We extend prior work,
which characterizes a real-time flow by its traffic pattern, delay
bound, timely-throughput requirement, and channel reliability, to
allow time-varying channels, allow clients to have different
deadlines, and allow for the optional employment of rate adaptation.
Thus, our model allow the treatment of more realistic fading
channels as well as scenarios with mobile nodes, and the usage of
more general transmission strategies.

We derive a sufficient condition for a scheduling policy to be
feasibility optimal, and thereby establish a class of feasibility
optimal policies. We demonstrate the utility of the identified class
by deriving a feasibility optimal policy for the scenario with rate
adaptation, time-varying channels, and heterogeneous delay bounds.
When rate adaptation is not available, we also derive a feasibility
optimal policy for time-varying channels. For the scenario where
rate adaptation is not available but clients have different delay
bounds, we describe a heuristic. Simulation results are also
presented which indicate the usefulness of the scheduling policies
for more realistic and complex scenarios.
\end{abstract}

\section{Introduction}  \label{section:introduction}

With the wide deployment of Wireless Local Area Networks (WLANs) and
advances in multimedia technology, wireless networks are
increasingly being used to carry real-time traffic, such as VoIP and
video streaming. These applications usually specify throughput
requirement while meeting specified delay bounds. We study the
problem of designing scheduling policies for such applications.

While there has been much research on scheduling real-time traffic
over wireline networks, the results are not directly applicable to
wireless networks where channels are unreliable, with qualities that
may be time-varying either due to fading or node mobility. Also,
individual clients may impose differing delay requirements. These
features present new challenges to the scheduling problems.

We consider the scenario where an Access Point (AP) is required to
serve real-time traffic for a set of clients. A previous work
\cite{IHH09MobiHoc} solves the scheduling problem in a restrictive
environment and proposes two feasibility optimal policies. In
particular, it assumes a fixed transmission rate, a static channel
model, and that all clients in the system require the same delay
bound. We extend this model so that it can capture the traffic
patterns, delay bounds, timely-throughput bounds, and delivery ratio
bounds of clients, for time-varying wireless channels. We address
scenarios with and without rate adaptation. We establish a
sufficient condition for a scheduling policy to be feasibility
optimal. Based on this we describe a class of policies and prove
that they are all feasibility optimal.

To demonstrate the utility of the class of policies, we study three
particular scenarios of interest. The first scenario employs rate
adaptation and treats time-varying channels, as well as allowing
different delay bounds for different clients. The other two
scenarios treat the case where rate adaptation is not available. One
scenario considers time-varying channels, while the other considers
the scenario where clients require different delay bounds. For the
former two scenarios, we derive computationally tractable scheduling
policies and prove that they are feasibility optimal. We also obtain
a heuristic for the third scenario.

We have also tested the derived policies using the IEEE 802.11
standard in a simulation environment. The results suggest that the
three policies outperform others, including the policies in
\cite{IHH09MobiHoc}, and a server-centric policy that schedules
packets randomly. In particular, since the policies introduced in
the previous work fail to provide satisfactory performance in the
environments studied here, this suggests that neglecting the facts
that the system can apply rate adaptation, that wireless channels
are time-varying, and the possibility that clients may require
different delay bounds, can result in malperformance of the derived
policies.

Section \ref{section:related work} reviews some of the related work.
Section \ref{section:system model} describes the extension of the
model in \cite{IHH09MobiHoc}. Section \ref{section:schedule}
discusses some useful observations for scheduling and reviews
policies proposed in \cite{IHH09MobiHoc}. In Section
\ref{subsection:extension for time varying channels}, we study an
extension for time-varying channels. In Section
\ref{section:sufficient condition}, we derive a general class of
policies that are feasibility optimal. Based on this class, we
obtain scheduling policies in Sections \ref{section:rate adaptation}
and \ref{section:time varying channels}, and a heuristic in Section
\ref{section:heterogeneous delay}, for different scenarios. In
Section \ref{section:simulation}, we discuss implementation issues
and simulation results. Section \ref{section:conclusion} concludes
the paper.

\section{Related Work}  \label{section:related work}

The problem of providing QoS over unreliable wireless channels has
received growing interest in recent years. Tassiulas and Ephremides
\cite{LT93} have considered the problem in a single-hop network by
assuming ON/OFF channels and derived a throughput-optimal policy.
Though the policy is unaware of packet delay, Neely \cite{MN08} has
shown that average packet delay is constant regardless of the
network size. Andrews et al \cite{MA01} have proposed another policy
that aims to improve packet delay. They have proved that their
policy is also throughput optimal but offer no theoretical bound on
packet delays. Liu, Wang, and Giannakis \cite{QL06} have used a
cross-layer approach to provide differentiated service for a variety
of classes of clients. Grilo, Macedo, and Nunes \cite{AG03} have
proposed a resource-allocation algorithm based on the expected
transmission time of each packet. Since the expected transmission
time may not be an accurate indication of the actual transmission
time, their work cannot provide provable delay guarantees.
Raghunathan et al \cite{VR08} and Shakkottai and Srikant \cite{SS02}
have both approached this problem by analytically demonstrating
algorithms to minimize the total number of expired packets in the
system. Their results, however, cannot provide differentiated
service to different clients. Hou, Borkar, and Kumar \cite{IHH09}
have studied the problem of providing QoS based on delay bounds and
delivery ratio requirements, and proposed two optimal policies under
some restrictive assumptions. Their work has been further extended
to deal with variable-bit-rate traffic \cite{IHH09MobiHoc}. In this
paper, we extend this work to more realistic scenarios, including
rate adaptation, time-varying channels and heterogeneous delay
bounds among clients. Fattah and Leung \cite{HF02} and Cao and Li
\cite{YC01} have surveyed other existing scheduling policies for
providing QoS.

\section{System Model}  \label{section:system model}
We begin by extending the model proposed in \cite{IHH09MobiHoc},
which only considers a static channel condition and fixed delay
bounds for all clients, to account for network behavior and
application requirements for providing QoS in wireless systems.

Consider a wireless system with $N$ clients, $\{1,2,\dots,N\}$, and
one access point (AP). Packets for clients arrive at the AP. Time is
slotted with slots $t\in \{0,1,2,\dots\}$. Time slots are further
grouped into \textit{periods} $[kT, (k+1)T)$ with period length $T$.
Packets arrive at the AP at the beginning of each period, at time
slots $\{0,T,2T,\dots\}$, probabilistically, with no more than one
packet per client. We model the packet arrivals as a stationary,
irreducible Markov process with finite state. The average
probability that packets arrive for subset $S$ of clients is $R(S)$.
Packet arrivals can be dependent between clients, and packet
arrivals in a period can depend on other periods.

Each client $n$ specifies a delay requirement $\tau_n$, with
$\tau_n\leq T$. If the packet for client $n$ is not delivered by the
${\tau_n}^{th}$ time slot of the period, the packet expires and is
discarded. This scheme applies naturally to a wide range of
server-centric wireless communication technologies, such as IEEE
802.11 Point Coordination Function (PCF), WiMax, and Bluetooth.

We consider an unreliable, heterogeneous, and time-varying channel
model. We model the channel condition as a stationary, irreducible
Markov process with a finite set of channel states $\mathcal{C}$.
The average probability that channel state $c$ occurs is $f_c$ and
the channel state remains constant within each period. We consider
the system both with rate adaptation and without. When rate
adaptation is not available, that is, when all packets are
transmitted at a fixed rate, the AP can make exactly one
transmission in each time slot. Under channel state $c$, the link
reliability between the AP and client $n$ is $p_{c,n}$, so that a
packet transmitted by the AP for client $n$ is delivered with
probability $p_{c,n}$. On the other hand, when the system uses rate
adaptation, the channel states describe the maximal rates that can
be supported between the AP and clients, which in turn decide the
service times for transmissions. Under channel state $c$, it takes
$s_{c,n}$ time slots to make an error free transmission to client
$n$.

The channel state and the packet arrivals in a period are assumed to
be independent of each other. We also assume that the AP has
knowledge of channel state, as well as whether a transmission is
successful, for example, through ACKs, in which case $p_{c,n}$ is
the probability that the AP receives an ACK after making a
transmission.

Each client $n$ requires a timely-throughput of at least $q_n$
packets per period. Since, on average, there are $\sum_{S:n\in S}
R(S)$ packets for client $n$ per period, this timely-throughput
bound can also be interpreted as a delivery ratio requirement of
$\frac{q_n}{\sum_{S:n\in S} R(S)}$.

\begin{definition}  \label{definition:fulfill}
A set of clients, $\{1,2,\dots,N\}$ is \emph{fulfilled} under a
scheduling policy $\eta$, if for every $\epsilon>0$,
\[
Prob\{\frac{d_n(t)}{t/T}>q_n-\epsilon\mbox{, for every
$n$}\}\rightarrow 1\mbox{, as } t\rightarrow\infty,
\]
where $d_n(t)$ is the number of packets delivered to client $n$ up
to time $t$.
\end{definition}

\section{Scheduling Policies}   \label{section:schedule}

Since the overall system can be viewed as a controlled Markov chain,
we have:

\begin{lemma}   \label{lemma:static policy}
For any set of clients that can be fulfilled, there exists a
stationary randomized policy that fulfills the clients, which uses a
probability distribution based only on the channel state, the set of
undelivered packets, and the number of time slots remaining in the
system (and not any events depending on past periods), according to
which it randomly chooses an undelivered packet to transmit, or
stays idle.
\end{lemma}

Since the computational overhead for some complex policies may be
too high for real-time applications, we consider the limited set of
\textit{priority-based policies}, which require computation only at
the beginning of each period:

\begin{definition}  \label{definition:priority-based policy}
A \emph{priority-based policy} is a scheduling policy which assigns
priorities to some of the clients, based on past history and current
state of the system, at the beginning of each period. During the
period, a packet for a client is transmitted only after all packets
for clients with higher priorities have been delivered. Packets for
clients which do not receive a priority are never transmitted. A
\emph{stationary randomized priority-based policy} is one which
chooses the priority order randomly according to a probability
distribution that depends only on the channel state and packet
arrivals at the beginning of each period. We denote by $\mathbb{P}$
and $\mathbb{P}_{rand}$ the sets of priority-based policies and
stationary randomized priority-based policies.
\end{definition}

\begin{definition}  \label{definition:feasible}
A set of clients is \emph{feasible in the set $\mathbb{P}$} (or
$\mathbb{P}_{rand}$) if there exists some scheduling policy in
$\mathbb{P}$ (or $\mathbb{P}_{rand}$) that fulfills it.
\end{definition}

Similar to Lemma \ref{lemma:static policy}, if $[q_n]$ is feasible
in the set $\mathbb{P}$, it is also feasible in the set
$\mathbb{P}_{rand}$.

\begin{definition}
We call the region in the $N$-space formed by vectors $[q_n]$ for
which the clients are feasible in $\mathbb{P}$ (or all policies), as
the \emph{feasible region under $\mathbb{P}$} (or all policies).
\end{definition}

\begin{lemma}   \label{lemma:convex feasible region}
The feasible region under the class of all policies, or
$\mathbb{P}$, are both convex sets.
\end{lemma}
\begin{proof}
Let $[q_n]$ and $[q'_n]$ be two vectors in the feasible region under
$\mathbb{P}$, and thus also feasible in $\mathbb{P}_{rand}$. Let
$\eta$ and $\eta'$ be policies in $\mathbb{P}_{rand}$ that fulfill
the two vectors, respectively. Then, the policy in $\mathbb{P}$ that
randomly picks one of the two policies, with $\eta$ being chosen
with probability $\alpha$, at the beginning of each period, fulfills
the vector $[\alpha q_n + (1-\alpha)q'_n]$. Further, since $q_n$ and
$q'_n$ are both larger than 0 for each $n$, $\alpha
q_n+(1-\alpha)q'_n>0$ for all $n$. Thus, the vector $[\alpha q_n +
(1-\alpha) q'_n]$ also falls in the feasible region under
$\mathbb{P}$. A similar proof holds for the class of all policies.
\end{proof}

Note that if $[q_n]$ is feasible in $\mathbb{P}$, then so is
$[q'_n]$, where $0<q_n'\leq q_n$.

\begin{definition}  \label{definition:strictly_feasible}
$[q_n]$ is \emph{strictly feasible in $\mathbb{P}$} (or the class of
all policies) if there exists some $\alpha\in(0,1)$ such that
$[q_n/\alpha]$ is feasible in $\mathbb{P}$ (or the class of all
policies).\footnote{Equivalently, $[q_n]$ is an interior point of
the feasible region under $\mathbb{P}$ (or the class of all
policies).}
\end{definition}

\begin{definition}  \label{definition:feasibility_optimal}
A scheduling policy $\eta$ is \emph{feasibility optimal among
$\mathbb{P}$} (or the class of all policies) if it fulfills every
set of clients that is strictly feasible in $\mathbb{P}$ (or the
class of all policies).
\end{definition}

In the rest of the paper, unless otherwise specified, the default is
the set of all policies.

\subsection{The Static Channel Case} \label{section:static channel}

In previous work \cite{IHH09MobiHoc}, the problem of admission
control and feasibility optimal scheduling has been addressed for
the case where the channel state is static, and all clients require
the same delay bounds, i.e. $|\mathcal{C}|=1$ and $\tau_n
\equiv\tau$. In the special case, we will use $p_n$ instead of
$p_{c,n}$ since the channel state is static, and $\tau$ instead of
$\tau_n$.

Two \emph{largest debt first} scheduling polices were proved to be
feasibility optimal, where the AP, based on the past history,
calculates a debt for each client. In each period, the AP sorts all
clients according to their debts, and schedules a packet for client
$n$ only after all packets for clients with larger debts have been
delivered. The first policy, the \textit{largest time-based debt
first policy}, uses the \textit{time-based debt} for client $n$ at
time slot $t$, defined as $\frac{t}{T}w_n$ minus the number of time
slots that the AP has spent on transmitting packets for client $n$
up to time slot $t$. The other policy, the \textit{largest
weighted-delivery debt first policy}, uses the
\textit{weighted-delivery debt} for client $n$ at time slot $t$,
defined as $\frac{\frac{t}{T}q_n-d_n(t)}{p_n}$, where $d_n(t)$ is
the number of delivered packets for client $n$ up to time slot $t$.

As for admission control, the following lemma was proved in
\cite{IHH09MobiHoc}:

\begin{lemma}   \label{lemma:workload}
A set of clients is fulfilled if and only if the long-term average
number of time slots that the AP spends on transmitting packets for
client $n$ per period is at least $w_n=\frac{q_n}{p_n}$ for each
$n$.
\end{lemma}

Further, since expired packets are dropped, the number of packets in
the system is bounded. Thus, there may be some time slots where the
AP may have delivered all packets in the system, and is therefore
forced to stay idle. For any subset $S$ of $\{1,2,\dots, N\}$,
define $I_S$ to be the minimum number of time slots that the AP is
idle in a period for any scheduling policy, given that the AP can
only transmit packets for the subset $S$ of clients. A necessary and
sufficient condition for strict feasibility is proved:

\begin{theorem} \label{theorem:necessary and sufficient}
A set of clients is strictly feasible if and only if $\sum_{n\in
S}w_n< T-E[I_S]$, for all $S\subseteq \{1,2\dots,N\}$.
\end{theorem}

\section{Time-Varying Channels}
\label{subsection:extension for time varying channels}

We now discuss how to extend the aforementioned policies to provide
QoS for time-varying channels. One intuitive approach is to decouple
the channel states. The AP assigns a timely-throughput requirement
$q_{c,n}$ for each channel state $c$ and client $n$, with
$\sum_{c\in \mathcal{C}} f_cq_{c,n}\geq q_n$. Also, for each channel
state $c$, the assigned throughput requirements must be strictly
feasible under that channel state, that is, $\sum_{n\in
S}\frac{q_{c,n}}{p_{c,n}}< T - E[I_{c,S}]$ for all $S\subseteq
\{1,2,\dots,N\}$ , where $I_{c,S}$ is the minimal number of time
slots that the AP is forced to stay idle in a period under channel
state $c$ for any scheduling policy, given that the AP only
transmits packets for the subset $S$ of clients. More formally, we
therefore seek a matrix $Q=[q_{c,n}]$ that solves the following
linear programming problem:
\begin{align*}
&\mbox{Max }\mbox{$\sum_{n=1}^N\sum_{c\in\mathcal{C}}$}f_cq_{c,n}\\
\mbox{s.t. } & \mbox{$\sum_{c\in\mathcal{C}}$} f_cq_{c,n} \geq q_n, \forall n \\
&\mbox{$\sum_{n\in S}$} \frac{q_{c,n}}{p_{c,n}} < T - E[I_{c,S}],
\forall c, \forall S\subseteq\{1,2,\cdots,N\}.
\end{align*}
After obtaining the matrix $Q$, we can modify the two largest debt
first policies to deal with time-varying channel conditions. Let
$s_c(t)$ be the number of time slots up to time slot $t$ that the
channel state has been $c$, and assume that the channel state at
time slot $t$ is $c$. In the largest time-based debt first policy,
we define the time-based debt for client $n$ under channel state $c$
as $\frac{s_c(t)}{T}\frac{q_{c,n}}{p_{c,n}}$ minus the number of
time slots that the AP has spent on transmitting packets for client
$n$ under channel state $c$ up to time slot $t$. In the largest
weighted-delivery debt first policy, we define the weighted-delivery
debt for client $n$ under channel state $c$ as
$\frac{\frac{s_c(t)}{T}q_{c,n}-d_{c,n}(t)}{p_{c,n}}$, where
$d_{c,n}(t)$ is the number of delivered packets for client $n$ under
channel state $c$. Obviously, these two modified largest debt first
policies are feasibility optimal.

While this extension offers feasibility optimality, the above linear
program involves exponentially many constraints. Further, it also
requires the knowledge of the distribution $[p_{c,n}]$ of channel
states. In many scenarios, such as those with mobile nodes, this
knowledge may not be available. This motivates us, in the following
sections, to describe a more general class of feasibility optimality
policies, and derive an on-line scheduling policy that is
feasibility optimal for the time-varying channel conditions.

\section{A Sufficient Condition for Feasibility Optimality}
\label{section:sufficient condition}

We now describe a more general class of policies that is feasibility
optimal. We start by extending the concept of ``debt''.

\begin{definition}  \label{definition:representing debt}
A variable $r_n(k)$, whose value is determined by the past history
of the client $n$ up to the $k^{th}$ period, or time slot $kT$, is
called a \emph{pseudo-debt} if:

\begin{enumerate}
\item $r_n(0) = 0$, for all $n$.

\item At the beginning of each period, $r_n(k)$ increases by a
constant strictly positive number $z_n=z_n(q_n)$, which is an
increasing linear function of $q_n$.

\item $r_n(k+1)=r_n(k)+z_n(q_n)-\mu_n(k)$, where $\mu_n(k)$ is
a non-negative and bounded random variable whose value is determined
by the behavior of client $n$. Further, $\mu_n(k)=0$ if the AP does
not transmit any packet for client $n$.

\item The set of clients is fulfilled if and only if
$Prob\{\frac{r_n(k)}{k} < \varepsilon\}\rightarrow 1$, as
$k\rightarrow\infty$, for all $n$ and all $\varepsilon>0$.
\end{enumerate}
\end{definition}

In the following example, we illustrate that both the time-based
debt and the weighted-delivery debt are pseudo-debts under a static
channel model.

\begin{example} \label{example:representing debt}
At the beginning of each period, the time-based debt $r^{(1)}_n(k)$
increases by $w_n=\frac{q_n}{p_n}$, and decreases by the number of
time slots that the AP has transmitted packets for client $n$ during
the period. Lemma \ref{lemma:workload} shows that condition (4) is
satisfied.

Similarly, $r^{(2)}_n(k)$, the weighted-delivery debt is also a
special case. It increases by $\frac{q_n}{p_n}$ at the beginning of
each period, and decreases by $\frac{1}{p_n}$ if a packet is
delivered for client $n$ during that period, and 0 otherwise. It
satisfies condition (4) by definition.   $\Box$
\end{example}

We can also define the \emph{feasible region for debt in
$\mathbb{P}$} (or in the set of all policies) as the set of $[z_n]$
such that the corresponding $[q_n]$ is feasible in $\mathbb{P}$  (or
in the set of all policies). Since $z_n$ is a linear function of
$q_n$ and the feasible region for $[q_n]$ is a convex set (Lemma
\ref{lemma:convex feasible region}), the feasible region for $[z_n]$
is also a convex set.

Using the concept of pseudo-debt, we prove a sufficient condition
for feasibility optimality. The proof resembles one used by Neely
\cite{MN08}, though in a different context, and is based on:

\begin{theorem}[Lyapunov Drift Theorem]\label{theorem:Lyapunov drift}
Let $L(t)$ be a non-negative Lyapunov function. Suppose there exists
some constant $B>0$ and non-negative function $f(t)$ adapted to the
past history of the system such that:
\begin{align*}
&E\{L(t+1)-L(t)|\mbox{history up to time $t$}\} \leq B-\epsilon
f(t),
\end{align*}
for all $t$, then:
$
\limsup_{t\rightarrow\infty}\frac{1}{t}\sum_{i=0}^t E\{f(i)\}\leq
B/\epsilon.
$
$\Box$
\end{theorem}

\begin{theorem} \label{theorem:sufficient}
Let $r_n(k)$ be a pseudo-debt.
\begin{enumerate}
\item A policy that maximizes the \emph{payoff function}
\begin{equation}\label{equation:payoff function}
\sum_{n=1}^N E\{r_n(k)^+\mu_{n}(k)|c_k, S_k, [r_m(k)]\}
\end{equation}
at the beginning of each period is feasibility optimal, where $c_k$
denotes the channel state in the $k^{th}$ period, and $S_k$ is the
subset of clients whose packets arrive at the AP at the beginning of
the $k^{th}$ period.
\item \label{item:two}A priority-based policy that maximizes (\ref{equation:payoff function})
over all policies in $\mathbb{P}$ is feasibility optimal in
$\mathbb{P}$.
\end{enumerate}
\end{theorem}
\begin{proof}
We present the proof for $\mathbb{P}$ only. A similar proof works
for the class of all policies too. Define $
L(k)=\frac{1}{2}\sum_{n=1}^Nr_n(k)^2.$ Since $r_n(k+1) =
r_n(k)+z_n-\mu_n(k)$,
\begin{align*}
&\Delta(L(k)):=E\{L(k+1)-L(k)|[r_m(k)]\}\\
=&E\{\frac{1}{2}\sum_{n=1}^Nr_n(k+1)^2-\frac{1}{2}\sum_{n=1}^Nr_n(k)^2|[r_m(k)]\}\\
=&E\{\sum_{n=1}^Nr_n(k)[z_n-\mu_n(k)]+\frac{1}{2}\sum_{n=1}^N[z_n-\mu_n(k)]^2|[r_m(k)]\}.
\end{align*}

Define
$B(k):=E\{\frac{1}{2}\sum_{n=1}^N[z_n-\mu_n(k)]^2|[r_m(k)]\}$. Then
$B(k)\leq B$, for all $k$, for some $B$. Hence for any policy in
$\mathbb{P}$:
\begin{equation}    \label{equation:Lyapunov drift}
\Delta(L(k))\leq E\{\sum_{n=1}^Nr_n(k)[z_n-\mu_n(k)]|[r_m(k)]\}+B.
\end{equation}

Suppose $[q_n]$ is strictly feasible in $\mathbb{P}$. The vector
$[z_n]$ is thus an interior point of the feasible region (for debt)
under $\mathbb{P}$, and there therefore exists some $\alpha\in(0,1)$
such that $[z_n/\alpha]$ is also in the feasible region under
$\mathbb{P}$. Let $z_{min}=\min\{z_1,z_2,\dots,z_N\}$. The
$N$-dimensional vector $[z_{min}]$ whose elements are all $z_{min}$,
falls in the feasible region under $\mathbb{P}$. Since the feasible
region under $\mathbb{P}$ is a convex set, the vector
$\alpha[z_n/\alpha]+(1-\alpha)[z_{min}]=[z_n+(1-\alpha)z_{min}]$ is
also in the feasible region under $\mathbb{P}$.

By Lemma \ref{lemma:static policy}, there exists a stationary
randomized policy $\eta'$ in $\mathbb{P}$ that fulfills the set of
clients with timely-throughput bounds for the vector
$[z_n+(1-\alpha)z_{min}]$. Let $\mu'_n(k)$ be the decrease in the
pseudo-debt for client $n$ under $\eta'$ during the period. Then, we
have:
\begin{align*}
E\{\mu'_n(k)|[r_m(k)]\}&=E\{E\{\mu'_n(k)|c_k,S_k,[r_m(k)]\}\}\\
&\geq z_n+(1-\alpha)z_{min}.
\end{align*}
Above, the outer expectation in the RHS is taken over channel states
and the vectors of packet arrivals.

Let $\eta$ be a policy that maximizes the payoff function
(\ref{equation:payoff function}),
for all $k$, among all policies in $\mathbb{P}$. Then defining
$\mu_n(k)$ and $r_n(k)$ as the decrease resulting from policy $\eta$
and the pseudo-debt, we have:
\begin{align*}
&\mbox{$\sum_{n=1}^N$} E\{r_n(k)^+\mu_{n}(k)|c_k,S_k,
[r_m(k)]\}\\\geq&\mbox{$\sum_{n=1}^N$}
E\{r_n(k)^+\mu'_{n}(k)|c_k,S_k, [r_m(k)]\}.
\end{align*}

We can assume without loss of generality that the policy does not
work on any client $n$ with $r_n(k)\leq 0$, that is, $\mu_n(k)=0$ if
$r_n(k)\leq 0$.\footnote{Since a policy cannot lose its feasibility
optimality by doing more work, this assumption is not restrictive.}
From (\ref{equation:Lyapunov drift}), we obtain:
\begin{align*}
\Delta(L(k))&\leq E\{\mbox{$\sum_{n=1}^N$}r_n(k)^+[z_n-\mu_n(k)]|[r_m(k)]\}+B\\
&\leq E\{\mbox{$\sum_{n=1}^N$}r_n(k)^+[z_n-\mu'_n(k)]|[r_m(k)]\}+B\\
&\leq -\mbox{$\sum_{n=1}^N$}r_n(k)^+(1-\alpha)z_{min}+B.
\end{align*}
Let $\epsilon:=(1-\alpha)z_{min}$. By Theorem \ref{theorem:Lyapunov
drift},
\begin{equation}    \label{equation:finally bounded}
\mbox{$\limsup_{k\rightarrow\infty}\frac{1}{k}\sum_{i=0}^{k}E\{\sum_{n=1}^Nr_n(k)^+\}\leq
B/\epsilon$}.
\end{equation}

Finally, since $z_n$ is a constant and $\mu_n(k)$ is a bounded
function, $|r_n(k+1)-r_n(k)|$ is bounded, which implies that
$|\sum_{n=1}^Nr_n(k+1)^+-\sum_{n=1}^Nr_n(k)^+|$ is also bounded for
all $k$. Thus, (\ref{equation:finally bounded}) implies that
$\frac{1}{k}E\{\sum_{n=1}^Nr_n(k)^+\}\rightarrow 0$ as
$k\rightarrow\infty$, as shown in Lemma \ref{lemma:bounded
convergence} below. This shows that $\frac{{r_n(k)}^+}{k}$ converges
to 0 in probability for all $n$. Hence, $\eta$ is feasibility
optimal in $\mathbb{P}$.
\end{proof}

\begin{lemma}   \label{lemma:bounded convergence}
Let $f(t)$ be a non-negative function such that $|f(t+1)-f(t)|\leq
M$, for some $M>0$, for all $t$. If
$\limsup_{t\rightarrow\infty}\frac{1}{t}\sum_{i=0}^tf(i)\leq
B/\epsilon,$ then $\lim_{t\rightarrow\infty}\frac{1}{t}f(t)=0.$
\end{lemma}
\begin{proof}
We prove by contradiction. Suppose
$\limsup_{t\rightarrow\infty}\frac{1}{t}f(t)>\delta$, for some
$\delta>0$. Thus, $f(t)>t\delta$ infinitely often. Suppose
$f(t)>t\delta$ for some $t$. Since $|f(t)-f(t-1)|<M$, we have
$f(t-1)>t\delta-M$. Similarly, $f(t-2)>t\delta-2M,$
$f(t-3)>t\delta-3M,\dots, f(t-\lfloor
t\delta/M\rfloor)>t\delta-\lfloor t\delta/M\rfloor M\geq0$. Summing
over these terms gives: $ \sum_{i=t-\lfloor t\delta/M\rfloor}^t
f(i)>\frac{t\delta\lfloor t\delta/M\rfloor}{2}, $ and thus, $
\sum_{i=0}^t \frac{1}{t}f(i)>\frac{\delta\lfloor
t\delta/M\rfloor}{2}. $ Since $f(t)>t\delta$ infinitely often,
$\limsup_{t\rightarrow\infty}\sum_{i=0}^t \frac{1}{t}f(i)=\infty,$
which is a contradiction.
\end{proof}

Theorem \ref{theorem:sufficient} suggests a more general procedure
to design feasibility optimal scheduling policies. To design a
scheduling policy in a particular scenario, we need to choose an
appropriate pseudo-debt and obtain a policy to maximize the payoff
function. Maximizing the payoff function is, however, in general,
difficult. Nevertheless, in some special cases, evaluating the
payoff function gives us simple feasibility optimal policies, or, at
least, some insights into designing a reasonable heuristic, as long
as we choose the correct pseudo-debt. In the following sections, we
demonstrate the utility of this approach.

\section{Scheduling Policy with Rate Adaptation}
\label{section:rate adaptation}

We now propose a feasibility optimal scheduling policy when rate
adaptation is employed. Channel qualities can be time-varying and
clients may have different deadlines.

To derive the scheduling policy, we define the \emph{delivery debt}
$r_n^{(3)}(k):=q_nk-d_n(kT),$ where $d_n(t)$ is the number of
delivered packets for client $n$ up to time slot $t$. Thus, $z_n
:=q_n$, while $\mu_n(k)=1$ if a packet for client $n$ is delivered
in the period, and $\mu_n(k)=0$ otherwise.

Suppose at the beginning of period $k$, the delivery debt vector is
$[r_n^{(3)}(k)]$, the channel state is $c$, and the set of arrived
packets is $S$. The transmission time for client $n$ is $s_{c,n}$
time slots, and client $n$ stipulates a delay bound of $\tau_n$.
Since transmissions are assumed to be error-free when rate
adaptation is applied, the scheduling policy consists of finding an
ordered subset $S' = \{m_1, m_2,\dots, m_{N'}\}$ of $S$ such that
$\sum_{n=1}^l s_{c,n} \leq \tau_l$, for all $1\leq l\leq m_{N'}$.
That is, when clients are scheduled according to the ordering, no
packets for clients in $S'$ would miss their respective delay
bounds. By Theorem \ref{theorem:sufficient}, a policy using an
ordered set $S'$ that maximizes $\sum_{n\in S'} r_n^{(3)}(k)$ with
the above constraint is feasibility optimal. This is a variation of
the knapsack problem. When $S'$ is selected, reordering clients in
$S'$ in an earliest-deadline-first fashion also allows all packets
to meet their respective delay bounds. Based on this observation, we
derive the feasibility optimal scheduling algorithm, the Modified
Knapsack Algorithm. Let $M[n,t]$ be the maximum debt a policy can
collect if only clients $1$ through $n$ can be scheduled and all
transmissions need to complete before time slot $t$. Thus,
$\max_{S'}\sum_{n\in S'} r_n^{(3)}(k)=M[N,T]$. Also, iteratively:
\[
M[n,t]=\left\{ \begin{array}{lr} M[n,t-1] &\mbox{if } t > \tau_n,\\
\max\{M[n-1,t],\\r_n^{(3)}(k)+M[n-1,t-s_{c,n}]\} &\mbox{otherwise,}
\end{array}\right.
\]
where $M[n-1,t]$ is the maximum debt can be collected when client
$n$ is not scheduled, and $r_n^{(3)}(k)+M[n-1,t-s_{c,n}]$ is that
when client $n$ is scheduled. The complexity of this algorithm is
$O(N\tau)$, and it is thus reasonably efficient.

\begin{algorithm}[h]
\caption{Modified Knapsack Policy} \label{alg:rate adaptation}
\begin{algorithmic}[1]
\FOR{$n=1$ to $N$}
\STATE $r_n^{(3)}(k)=q_nk-d_n(kT)$\\
\ENDFOR
\STATE Sort clients such that $\tau_1\leq \tau_2\leq \dots\leq \tau_N$\\
\STATE $S'[0,0] = \phi$\\
\STATE $M[0,0] = 0$\\
\FOR{$n=1$ to $N$} \FOR{$t=1$ to $T$}\IF{$t>\tau_n$}
\STATE $M[n,t] = M[n,t-1]$\\
\STATE $S'[n,t] = S'[n,t-1]$\\
\ELSIF{client $n$ has a packet AND\\ $r_n^{(3)}(k)+M[n-1,t-s_{c,n}]
> M[n-1,t]$}
\STATE $M[n,t] = r_n^{(3)}(k)+M[n-1,t-s_{c,n}]$\\
\STATE $S'[n,t] = S'[n-1,t-s_{c,n}]+\{n\}$\\
\ELSE
\STATE $M[n,t] = M[n-1,t]$\\
\STATE $S'[n,t] = S'[n-1,t]$\\
\ENDIF\ENDFOR \ENDFOR
\STATE schedule according to $S'[N,T]$\\
\end{algorithmic}
\end{algorithm}

\section{Computationally Tractable Scheduling for Time-Varying Channels}
\label{section:time varying channels}

We now consider the case when rate adaptation is not available, and
propose a scheduling policy for time-varying channels and
homogeneous delay bounds. We show that the policy is feasibility
optimal among all priority-based policies. We use the delivery debt,
$r_n^{(3)}(k)$, of Section \ref{section:rate adaptation}.

Suppose at the beginning of a period, the delivery debt vector is
$[r_n^{(3)}(k)]$, the channel state is $c$, and the set of arrived
packets is $S$. We wish to find the priority ordering that maximizes
the payoff function
$\mu_{tot}(k)=\sum_{n=1}^Nr_n^{(3)}(k)^+E\{\mu_n(k)\}$, where in the
expectation we suppose that the channel state $c$ and the set of
arrival packets $S$ are both fixed. Obviously, transmitting a packet
from a client $n$ with $r_n^{(3)}(k)\leq 0$ will not increase the
value of $\mu_{tot}(k)$. Thus, we do not give priorities to clients
with non-positive delivery debts. For ease of the remaining
discussion, we further assume $r_n^{(3)}(k)>0$ for all $n$.

Consider two orderings, $A$ and $B$: In $A$, the priority order is
$\{1,2,\dots,N\}$, while, in $B$, the priority order is
$\{1,2,\dots,m-1,m+1,m,m+2,m+3,\dots,N\}$. Let the values of the
payoff functions be $\mu_{tot}^A$ and $\mu_{tot}^B$. Since clients
$1$ through $m-1$ have the same priorities in both orderings and
their priorities are higher than the remaining clients, the values
of $E\{\mu_n(k)\},$ $1\leq n\leq m-1$ are the same for both
orderings. On the other hand, clients $m+2$ through $N$ also have
the same priorities in both orderings and they can be scheduled only
after the packets for clients $1$ through $m+1$ are delivered. The
probabilities of packet deliveries for these clients are the same
under the two orderings. Thus, to compare the two orderings, one
only needs to evaluate the probabilities of packet delivery for
client $m$ and $m+1$. We further notice that the probabilities that
packets for both clients $m$ and $m+1$ are delivered are also the
same for both orderings. With $e_{n}$ the event that the packet for
client $n$ is delivered,
\begin{align*}
\mu_{tot}^A-\mu_{tot}^B=&r_m^{(3)}(k)Prob\{e_m\backslash
e_{m+1}|\mbox{ordering
A}\}\\&-r_{m+1}^{(3)}(k)Prob\{e_{m+1}\backslash e_{m}|\mbox{ordering
B}\}.
\end{align*}

Suppose that there are $\tau'$ time slots left when all packets from
client $1$ through $m-1$ have been delivered. The probability
distribution of $\tau'$ is the same under both orderings. Since the
channel reliability is $p_{c,n}$,
\begin{align*}
&\mu_{tot}^A-\mu_{tot}^B\\
=&\mbox{$r_m^{(3)}(k)E\{\sum_{t=1}^{\tau'}p_{c,m}(1-p_{c,m})^{t-1}(1-p_{c,m+1})^{\tau'-t}\}$}\\
&\mbox{$-r_{m+1}^{(3)}(k)E\{\sum_{t=1}^{\tau'}p_{c,m+1}(1-p_{c,m+1})^{t-1}(1-p_{c,m})^{\tau'-t}\}$}\\
=&\mbox{$[r_m^{(3)}(k)p_{c,m}-r_{m+1}^{(3)}(k)p_{c,m+1}]$}\\&\times
\mbox{$E\{\sum_{t=0}^{\tau'-1}(1-p_{c,m})^{t}(1-p_{c,m+1})^{\tau'-t-1}\}.$}
\end{align*}

Thus, $\mu_{tot}^A\geq\mu_{tot}^B$ if $r_m^{(3)}(k)p_{c,m}\geq
r_{m+1}^{(3)}(k)p_{c,m+1}$. This leads us to obtain the Joint
Debt-Channel Policy. The computation time is only $O(N\log N)$.

\begin{algorithm}[h]
\caption{Joint Debt-Channel Policy} \label{alg:time varying}
\begin{algorithmic}[1]
\FOR{$n=1$ to $N$}
\STATE $r_n^{(3)}(k)=q_nk-d_n(kT)$, for all $n$\\
\ENDFOR
\STATE Sort clients with a packet arrival such that $r_1^{(3)}(k)p_{c,1}\geq r_2^{(3)}(k)p_{c,2}\geq\dots\geq r_{N_0}^{(3)}(k)p_{c,N_0}> 0\geq r_{N_0+1}^{(3)}(k)p_{c,N_0+1}\geq\dots$\\
\STATE Transmit packets for clients $1$ through $N_0$ by the ordering\\
\end{algorithmic}
\end{algorithm}

\begin{theorem}
The joint debt-channel policy is feasibility optimal among all
priority-based policies.
\end{theorem}
\begin{proof}
Let $\eta$ be the joint debt-channel policy and $\eta'$ any
priority-based policy. Suppose the priorities assigned by the
policies are $\eta_1, \eta_2,\dots,\eta_m$, and
$\eta'_1,\eta'_2,\dots,\eta'_{m'}$. We modify $\eta'$ as follows:
\begin{enumerate}
\item Delete any element in $\eta'_1\sim\eta'_{m'}$ with
$r^{(3)}_{\eta'_n}(k)\leq 0$.
\item For any client $n$ with $r^{(3)}_n(k)>0$ that is not in
$\eta'_1\sim\eta'_{m'}$, append it at the end of the ordering.
\item If $\eta'_1\sim\eta'_{m'}$ is still different from
$\eta_1\sim\eta_m$, there exists some $n$ such that
$r_{\eta'_n}^{(3)}(k)p_{c,{\eta'_n}}<r_{\eta'_{n+1}}^{(3)}(k)p_{c,{\eta'_{n+1}}}$.
Swap $\eta'_n$ and $\eta'_{n+1}$.
\item Repeat Step 3 until the two orderings are the same.
\end{enumerate}

Steps 1 and 2 will not decrease the value of the payoff function. As
derived above, Step 3 does not decrease the value of the payoff
function, either. Thus, $\eta$ maximizes the payoff function and is
feasibility optimal in $\mathbb{P}$.
\end{proof}

\section{A Heuristic for Heterogeneous Delay Bounds}
\label{section:heterogeneous delay}

We now describe a heuristic for packet scheduling, for the case
where each channel state is static and transmission rate is fixed,
but clients require different delay bounds. We use $p_n$ to
represent channel reliability.

We will use the time-based debt, $r_n^{(1)}(k)$, as discussed in
Example \ref{example:representing debt}. The payoff function is
$E\{\sum_{n=1}^N r^{(1)}_n(k)^+\mu_n(k)\}$.

Suppose, without loss of generality, that at the beginning of a
period, packets for clients $\{1,2,\dots,N_0\}$ arrive. We further
assume that $\tau_1\leq\tau_2\leq\dots\leq \tau_{N_0}$. Let
$\gamma_n$ be the number of transmissions the AP needs to make for
client $n$ for success. While $\gamma_n$ is a random variable that
cannot be foretold, we examine how to maximize
$\sum_{n=1}^{N_0}r_n^{(1)}(k)^+\mu_n(k)$ if we \textit{knew}
$\gamma_n$.

We solve this by proceeding backwards in time. During time slots
$[\tau_{N_0-1}+1, \tau_{N_0}]$, all packets except the one for
client $N_0$ have expired, and we can only make transmissions for
client $N_0$ during these time slots. Thus, it does not make sense
to schedule client $N_0$ for more than
$\gamma^{N_0-1}_{N_0}:=\gamma_{N_0}-(\tau_{N_0}-\tau_{N_0-1})$
transmissions before time slot $\tau_{N_0-1}$. Next, in the time
slots between $[\tau_{N_0-2}+1, \tau_{N_0-1}]$, only clients $N_0-1$
and $N_0$ can be scheduled. An obvious choice is to schedule the
client with larger debt first, with the restriction that it is not
scheduled for more than $\gamma^{N_0-1}_n$ time slots, and to then
schedule the other client. (For simplicity, we let
$\gamma^{N_0-1}_{N_0-1}:=\gamma_{N_0-1}$.) We can further obtain the
remaining transmissions allowed for client $n$ before time slot
$\tau_{N_0-2}$, which we call $\gamma^{N_0-2}_n$, as
$\gamma^{N_0-1}_n$ minus the number of transmissions scheduled for
client $n$ during time slots $[\tau_{N_0-2}+1,\tau_{N_0-1}]$.
Transmissions of the remaining time slots are scheduled similarly.

While it is impossible to know the exact value of $\gamma_n$ in
advance, we can estimate it. One estimate is its expected value,
$\frac{1}{p_n}$. However, this estimate does not consider the
timely-throughput requirements. If a client has significantly larger
debt than others, a reasonably good policy would allocate enough
time slots so that the probability of packet delivery for the client
in this period is at least its delivery ratio bound,
$\frac{q_n}{\sum_{n\in S}R(S)}$, given that a packet for client $n$
arrived. So we estimate $\gamma_n$ by the number of transmissions
that we need to allocate for client $n$ so that it can achieve its
delivery ratio bound. Since the channel reliability for client $n$
is $p_n$, this estimate $\gamma_n$ is
$\lceil\log_{1-p_n}(1-\frac{q_n}{\sum_{n\in S}R(S)})\rceil$. We thus
derive the Adaptive-Allocation Policy shown in Algorithm
\ref{alg:deadlines}.

As a final remark, note that in all the three policies discussed in
this paper, we do not schedule transmissions for clients with
non-positive debts. This restriction improves the performance for
clients with non-real time traffic. In practice, it is possible that
clients with real-time traffic and clients with non-real time
traffic coexist. Thus, it is important not to allocate too much of
the resource to real-time clients and starve those with non-real
time traffic.

\begin{algorithm}[h]
\caption{Adaptive-Allocation Policy} \label{alg:deadlines}
\begin{algorithmic}[1]
\FOR{$n=1$ to $N$}
\STATE $r_n^{(1)}(k)=$ time-based debt\\
\STATE $\gamma_n = \lceil\log_{1-p_n}(1-\frac{q_n}{\sum_{n\in
S}R(S)})\rceil$\\
\ENDFOR
\STATE Sort clients so that packets for clients $1\sim N_0$ arrive and $r_1^{(1)}(k)\geq r_2^{(1)}(k)\geq\dots\geq r_{N_0}^{(1)}(k)$\\
\STATE $alloc\leftarrow n\times1-$vector\\
\FOR{$t=T$ to $1$}
\STATE $n\leftarrow 1$\\
\WHILE{($\tau_n>t$ or $\gamma_n\leq 0$) and $n\leq N_0$}
\STATE $n\leftarrow n+1$\\
\ENDWHILE \IF{$r_n^{(1)}(k)>0$}
\STATE $alloc[t]\leftarrow n$\\
\ELSE \STATE $alloc[t]\leftarrow N_0+1$\\
\ENDIF \IF{$n\leq N_0$}
\STATE $\gamma_n\leftarrow \gamma_n-1$\\
\ENDIF \ENDFOR \FOR{each time slot $t$} \IF{$alloc[t]\leq N_0$ and
the packet for client $alloc[t]$ has not been delivered}
\STATE transmit the packet for client $alloc[t]$\\
\ELSE
\STATE transmit the packet with the largest positive time-based debt\\
\ENDIF \ENDFOR
\end{algorithmic}
\end{algorithm}

\section{Simulation Results}    \label{section:simulation}

We have implemented the scheduling policies discussed in previous
sections by using the IEEE 802.11 PCF standard in the \textit{ns-2}
simulator. We present the simulation results for the scenario with
time-varying channels, and with clients requiring different delay
bounds. In each scenario, we compare our policies against the two
largest debt first policies of \cite{IHH09MobiHoc}, and a policy
that assigns priorities to clients randomly, \emph{random}. IEEE
802.11e, an enhancement to 802.11 for QoS, allows clients with
real-time traffic to use smaller contention window and inter frame
space to obtain priorities over clients with non-real time traffic.
However, clients with real-time traffic have to compete with each
other in a random access manner with equal channel access
probabilities, without any QoS based preference or discrimination.
Further, the inter frame space and contention window size are
smaller in PCF than in 802.11e. Thus, the random policy can be
viewed as an improved version of 802.11e. Similar to the previous
work, we conduct two sets of simulations for each scenario, one with
clients carrying VoIP traffic, and one with clients carrying video
streaming traffic. The major difference between the two settings
lies in their traffic patterns. Many VoIP codecs generate packets
periodically. Thus, future packet arrivals can be easily predicted
and may be dependent among different clients. For example, if two
clients generate packets at the same rate, then either all or none
of their packets arrive simultaneously. On the other hand, video
streaming technology, such as MPEG, may generate traffic with
variable-bit-rate (VBR). Thus, packets arrive at the AP
probabilistically, with probability depending on the context of the
current frame, and arrivals are independent among different clients.

For the VoIP traffic, we follow the standards of the ITU-T G.729.1
\cite{G729} and G.711 \cite{G711} codecs. Both codecs generate
traffic periodically. G.729.1 generates traffic with bit rates 8 --
32 kbits/s, while G.711 generates traffic at a higher rate of 64
kbits/s. We assume the period length, $T$, is 20 ms, and the payload
size of a packet is 160 Bytes. The codecs generate one packet every
several periods; with the duration between packet arrivals depending
on the bit rate used.

We use MPEG for the video streaming setting. MPEG VBR traffic is
usually modeled as a Markov chain consisting of three activity
states \cite{LJ99}\cite{IVMF04}. Each state generates traffic
probabilistically at different mean rates, with the state being
determined by the current frame of the video. The statistical mean
rates in each state are those obtained in an experimental study
\cite{IVMF04}. We use them in setting the traffic patterns of MPEG
traffic. We assume the period length to be 6 ms and the payload size
of a packet to be 1500 Bytes. Table \ref{table:mpeg traffic} shows
the statistical results of the experimental study \cite{IVMF04},
where we also present them in terms of the packet arrival
probability of our setting. In Table \ref{table:mpeg traffic},
``Data rate'' is measured in bits/GoP, where 1 GoP= 240 ms.

\begin{table}[t]
  \centering
  \caption{MPEG Traffic Pattern}\label{table:mpeg traffic}
  \begin{tabular}{|l|l|l|l|}
    \hline
    Activity & Great & High & Regular \\
    \hline
    Data rate & 501597 & 392237 & 366587 \\
    \hline
    Arrival probability & 1 & 0.8 & 0.75 \\
    \hline
  \end{tabular}
\end{table}

We simulate 20 runs for each setting, each run lasting one minute in
simulated time. All results shown are averaged over the 20 runs. A
natural performance metric for a client is the delivery debt,
$r_n^{(3)}(k)$. The performance of the system is measured by the sum
of the positive delivery debts of the clients, that is,
$\sum_{n=1}^N r_n^{(3)}(k)^+$,  the \textit{total delivery debt}. In
addition to evaluating how well the tested policies serve clients
with real-time traffic, we also wish to know whether the policies
starve those with non-real time traffic. Hence we add a client with
saturated non-real time traffic in all simulations. Packets for the
non-real time client are scheduled in all time slots that are left
idle otherwise. We measure the throughput of the client with
non-real time traffic by the average number of packets delivered.

\subsection{Rate Adaptation}

We present the simulation results under the scenario where rate
adaptation is applied, channels are time-varying, and clients may
require different delay bounds.

We first show the results for VoIP traffic. We use IEEE 802.11b as
the MAC protocol, which can provide a maximum data rate of 11 Mb/s.
We assume that the channel capacity of each client alternates
between 11 Mb/s and 5.5 Mb/s. Simulation results suggest that the
times needed for a transmission, including all MAC overheads such as
the time for waiting an ACK, are around 480 $\mu$s and 610 $\mu$s
for the two transmission rates, respectively. Ideally, the length of
a time slot should be a common divisor of the transmission times
needed under the two used data rates. We approximate this value by
160 $\mu$s. Thus, transmitting a packet requires 3 time slots when
using 11 Mb/s and 4 time slots when using 5.5 Mb/s. Further, a
period consists of 125 time slots.

There are two groups of clients, $A$ and $B$. Clients in group $A$
generate one packet every three periods, or at rate 21.3 kbits/s,
and require 90$\%$ of each of the clients' packets to be delivered,
or a timely-throughput requirement of 19.2 kbits/s. Clients in group
$B$ generate one packet every two periods at rate 32 kbits/s, and
require 70$\%$ of each of the clients' packets to be delivered,
corresponding to a timely-throughput requirement of 22.4 kbits/s.
The two groups can be further divided into subgroups, $A_1$, $A_2$,
$A_3$, $B_1$, and $B_2$, each with 22 clients. Clients in subgroup
$A_i$ generate packets at periods $[i,i+3,i+6,\dots]$, and clients
in subgroup $B_i$ generate packets at periods $[i,i+2,i+4,\dots]$.
Finally, clients in group $A$ require a delay bound equal to the
period length, or 125 time slots, while clients in group $B$ require
a delay bound equal to two-third of the period length, or 83 time
slots.

Simulation results are shown in Figure \ref{fig:audio_rate}. The
modified knapsack policy incurs the least total delivery debt among
all evaluated policies. This is because all the other three policies
neglect the time-varying channels with different data rates and the
heterogeneous delay bounds. Further, by only scheduling those
clients with positive delivery debts, the modified knapsack policy
achieves higher throughput for the non-real time client than both
the policies proposed in \cite{IHH09MobiHoc}. The random policy
results in the highest throughput for the non-real time client.
However, this is because it sacrifices the real-time clients. In
fact, its total delivery debt is more than 300 times larger than the
total delivery debt of the modified backpack policy. This huge
difference suggests that the random policy, and therefore also
802.11e, are not adequate for providing QoS when multiple clients
with real-time traffic are present.

\begin{figure}[t]\subfloat{
\label{fig:audio_channel_debt} 
\includegraphics[width=1.6in]{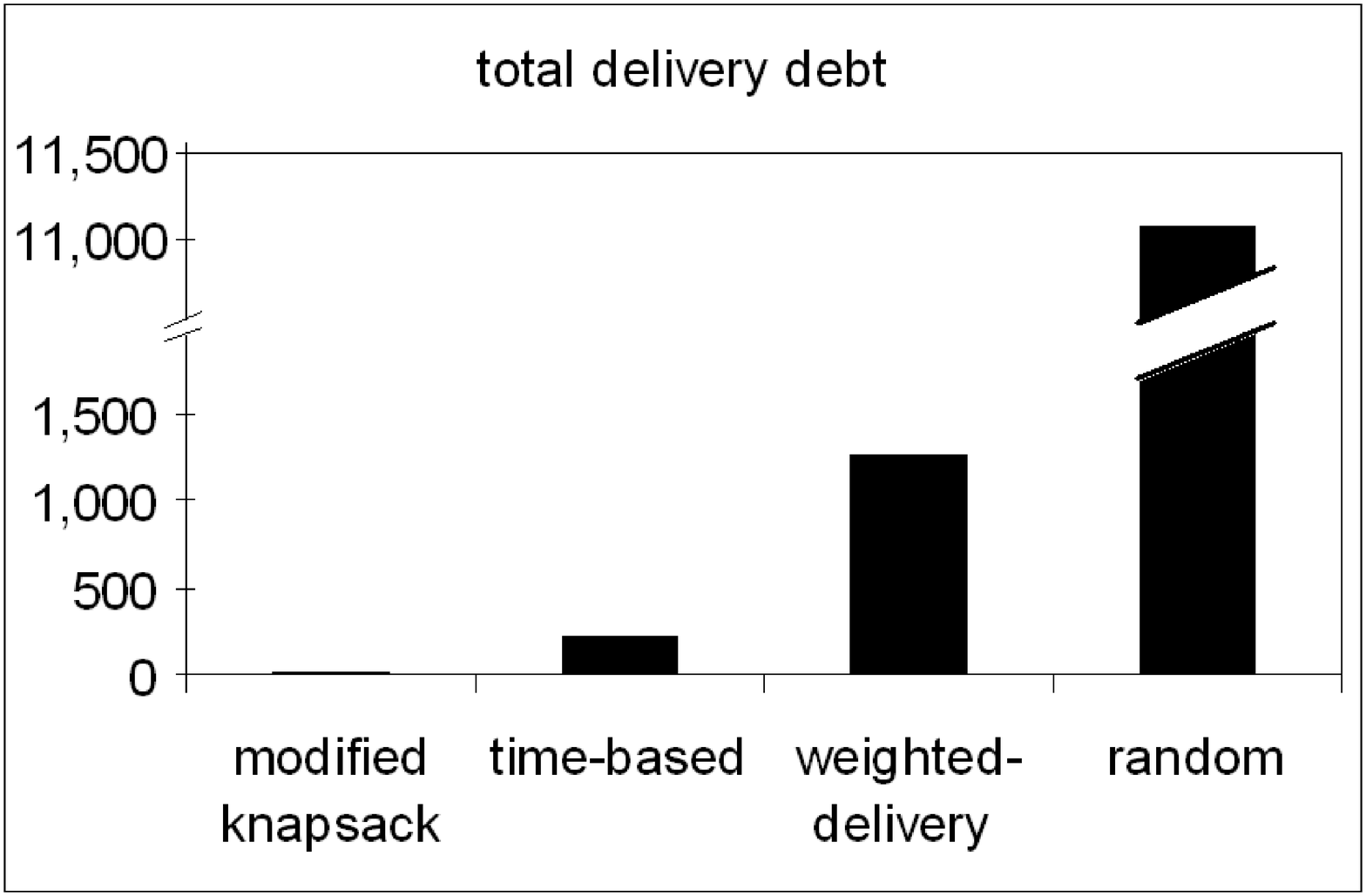}}
\hspace{0.01\linewidth}\subfloat{
\label{fig:audio_channel_throughput} 
\includegraphics[width=1.6in]{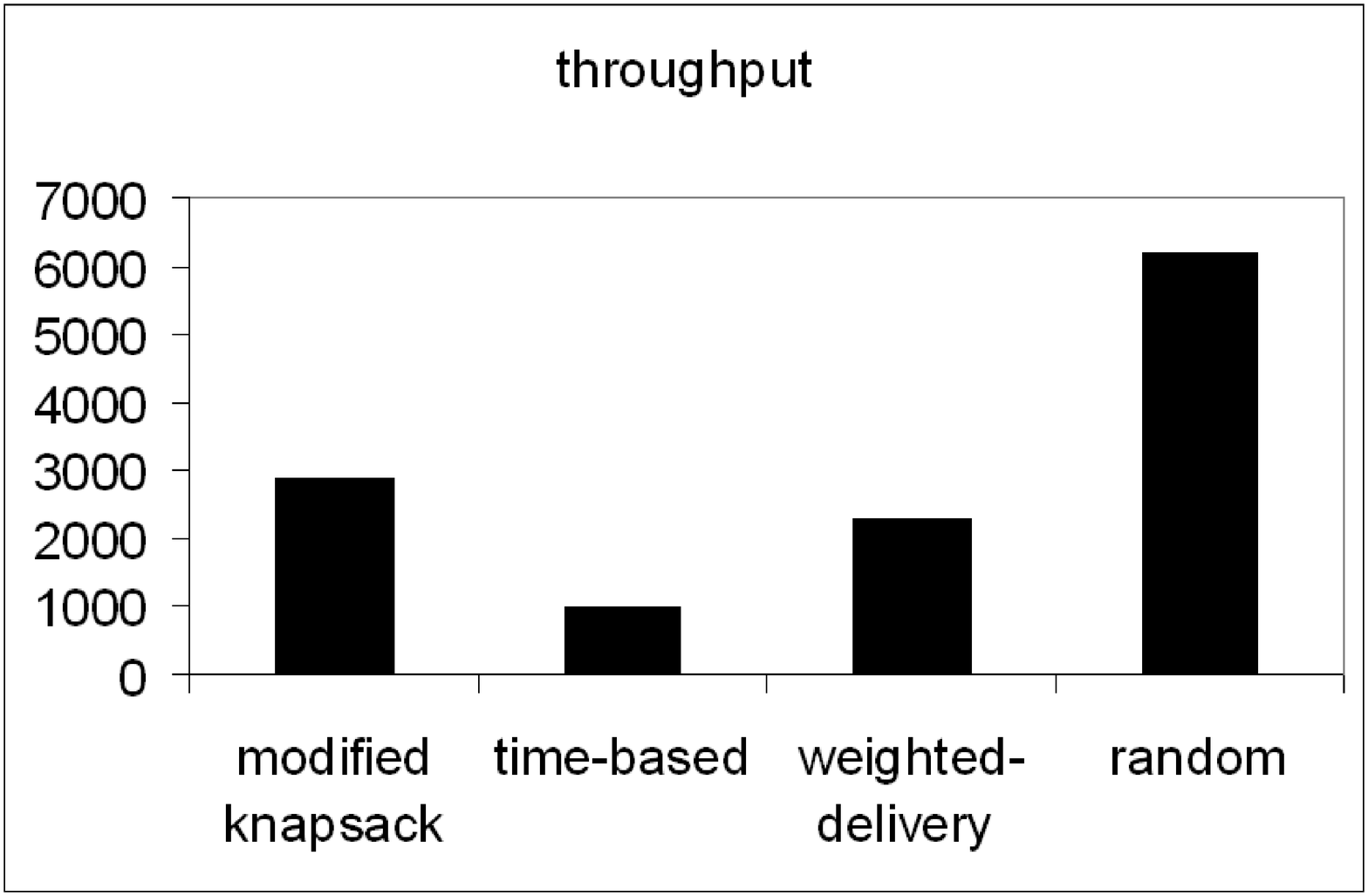}}
\caption{Performance for VoIP traffic with rate
adaptation.}\label{fig:audio_rate}
\end{figure}

Next we consider the scenario with MPEG traffic. Since video
streaming requires much higher bandwidth than VoIP, we use 802.11a
as the underlying MAC, which can support up to 54 Mb/s. We assume
that channel capacity for each client alternates between 54 Mb/s and
24 Mb/s. The transmission times for a data-ACK handshake require 660
$\mu$s with 54 Mb/s data rate, and 940 $\mu$s with 24 Mb/s. The
length of a time slot is 60 $\mu$s. Thus, the transmission times for
the two data rates are 11 time slots and 16 time slots,
respectively. Further, a period consists of 100 time slots.

We again assume there are two groups of clients. Clients in group
$A$ generate packets according to Table \ref{table:mpeg traffic},
and clients in group $B$ are assumed to offer only lower quality
video by generating packets only 80$\%$ as often as those in group
$A$, in each of the three states. We assume clients in group $A$
require 90$\%$ delivery ratios, and clients in group $B$ require
60$\%$ delivery ratios. Since the length of a period for MPEG is
very small, it is less meaningful to discuss heterogeneous delay
bounds. Thus, we assume all clients require a delay bound equal to
the length of a period. We further assume that there are 6 clients
in both groups.

Simulation results are shown in Figure \ref{fig:video_rate}. As in
the case of VoIP traffic, the modified knapsack policy achieves the
smallest total delivery debt among all the four policies. Also, by
not scheduling clients with non-positive debts, the modified
backpack policy also achieves the highest throughput for the
non-real time client.

\begin{figure}[t]\subfloat{
\label{fig:audio_channel_debt} 
\includegraphics[width=1.6in]{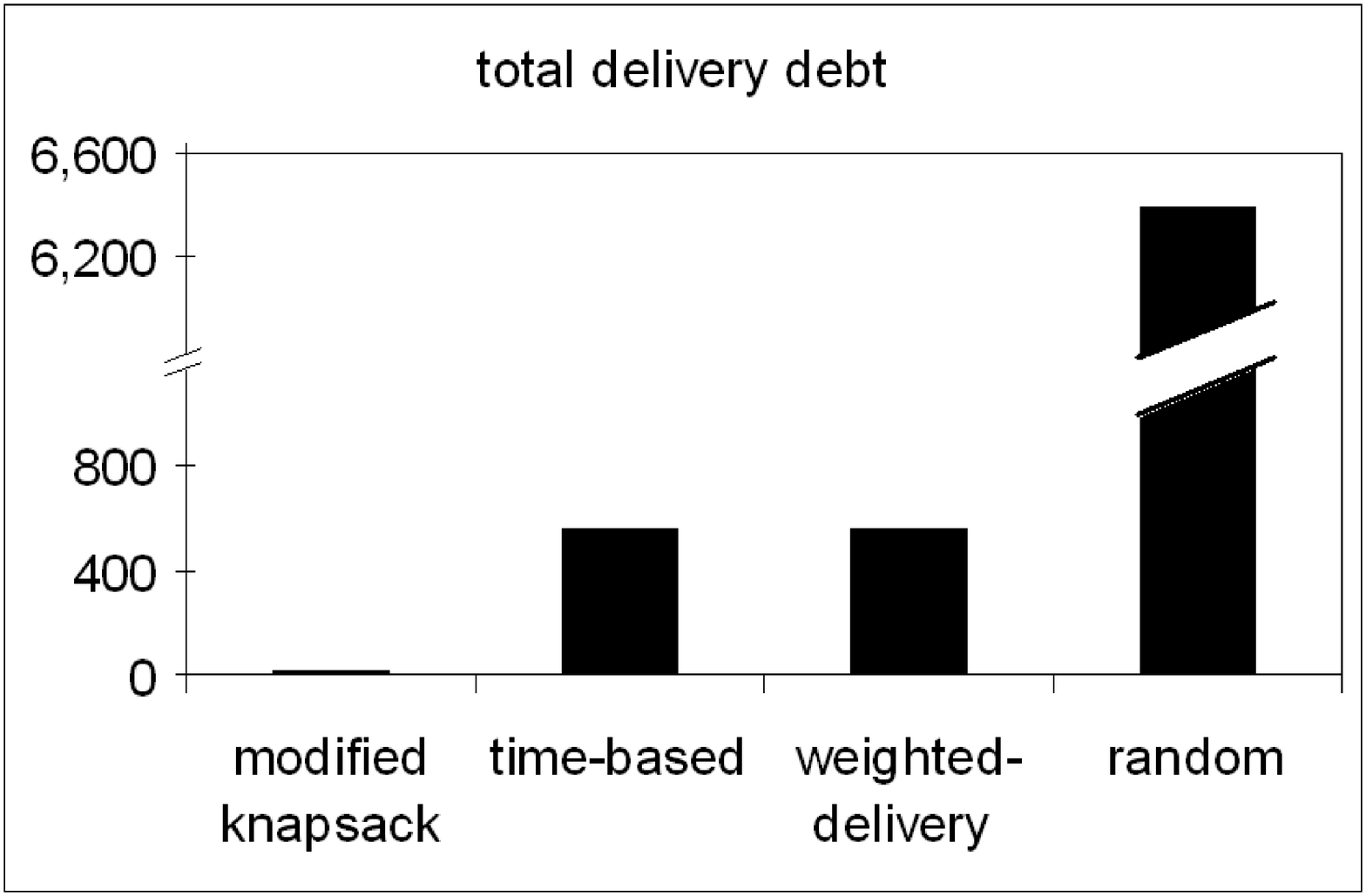}}
\hspace{0.01\linewidth}\subfloat{
\label{fig:audio_channel_throughput} 
\includegraphics[width=1.6in]{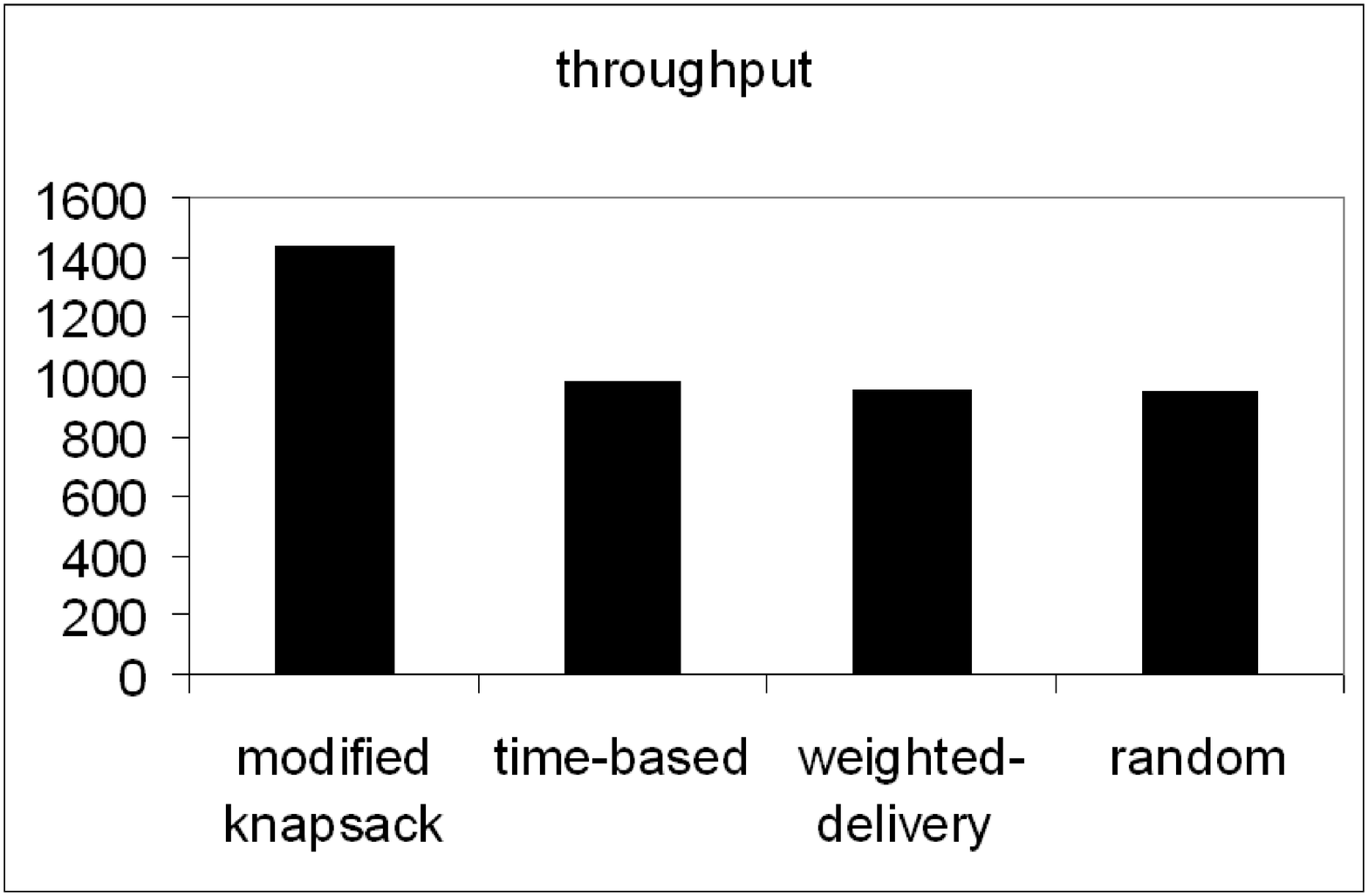}}
\caption{Performance for MPEG traffic with rate
adaptation.}\label{fig:video_rate}
\end{figure}

\subsection{Time-varying Channels}
We now consider the scenario with time-varying channels, with all
clients requiring delay bounds equal to period length. We model the
wireless channel by the widely used Gilbert-Elliot model
\cite{EOE63}\cite{ENG60}\cite{HSW95}, with the wireless channel
considered as a two-state Markov chain, with ``good'' state and
``bad'' states. A simulation study by Bhagwat et al \cite{PB97}
shows that the link reliability can be modeled as 100$\%$ when the
channel is in the good state, and 20$\%$ when the channel is in the
bad state. The duration that the channel stays in one state is
exponentially distributed with mean 1 -- 10 sec for the good state,
and 50 -- 500 msec for the bad state.

While modifying the two largest debt first policies as suggested in
Section \ref{subsection:extension for time varying channels} will
yield feasibility optimality, such modification requires solving the
linear programming problem and is intractable. Rather, we consider
some easier modifications for the two policies. For the largest
time-based debt first policy, we modify it so that it treats the
channel as a static one, with link reliability equal to the
time-averaged link reliability. For the largest weighted-delivery
debt first policy, the weighted-delivery debt for client $n$ at time
slot $t$ is defined as $\frac{t}{T}q_n-d_n(t)$ divided by the
current link reliability.

For the case of VoIP traffic, we use 802.11b as the underlying MAC
and use a fixed transmission rate of 11 Mb/s. We consider the same
two groups of clients as in the previous section. We assume that the
mean duration of the bad state is 500 msec for all clients, and the
mean duration of the good state is $1+0.5n$ sec for the $n^{th}$
client in each subgroup. The time-average link reliability of the
$n^{th}$ client in each subgroup can be computed as
$\frac{2.2+n}{3+n}$. There are 19 clients in each of the subgroups.

Simulations results are shown in Figure \ref{fig:audio_channel}. The
joint debt-channel policy incurs near zero total delivery debt,
while all the other policies have much larger total delivery debts.
The fact that the largest time-based debt first policy fails to
fulfill the set of clients suggests that only considering the
average channel reliability, without taking channel dynamics into
account, is not satisfactory. A somewhat surprising result is that
the total delivery debt for the largest weighted-delivery debt first
policy is even larger than that for the random policy. This is
because the policy favors those clients with poor channels. When the
channel state is time-varying, it may make more sense to postpone
the transmissions for a client with a poor channel until its channel
condition turns better. Thus, using weighted-delivery debt for
time-varying channels is not only inaccurate, but even harmful in
some settings. It can also be shown that the throughput for the
client with saturated non-real time traffic is the highest with the
joint debt-channel policy. By only scheduling those real-time
clients with positive delivery debts, the policy prevents putting
too much effort into any real-time client, and thus reserves enough
resources for clients with non-realtime traffic.

\begin{figure}[t]\subfloat{
\label{fig:audio_channel_debt} 
\includegraphics[width=1.6in]{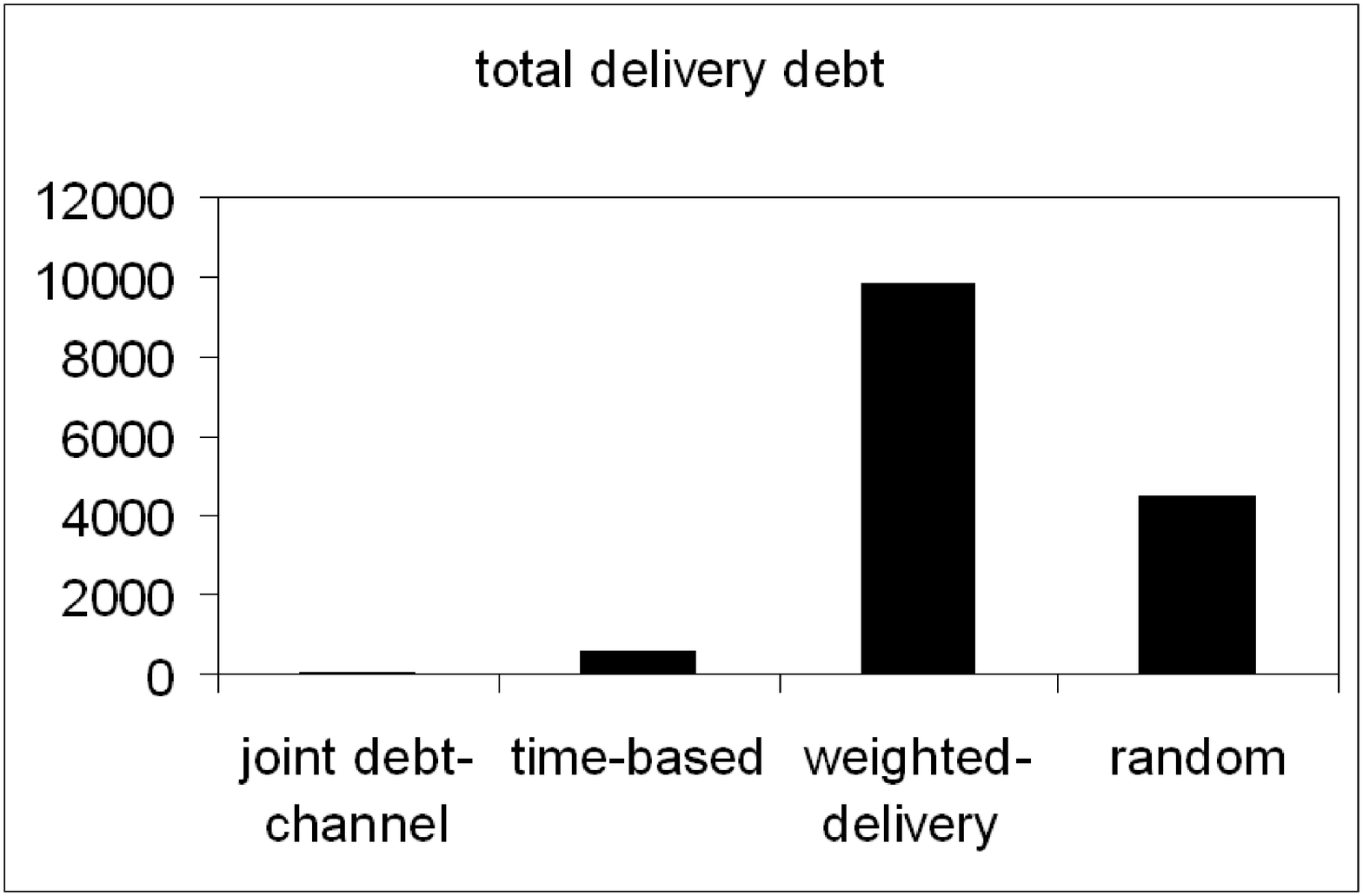}}
\hspace{0.01\linewidth}\subfloat{
\label{fig:audio_channel_throughput} 
\includegraphics[width=1.6in]{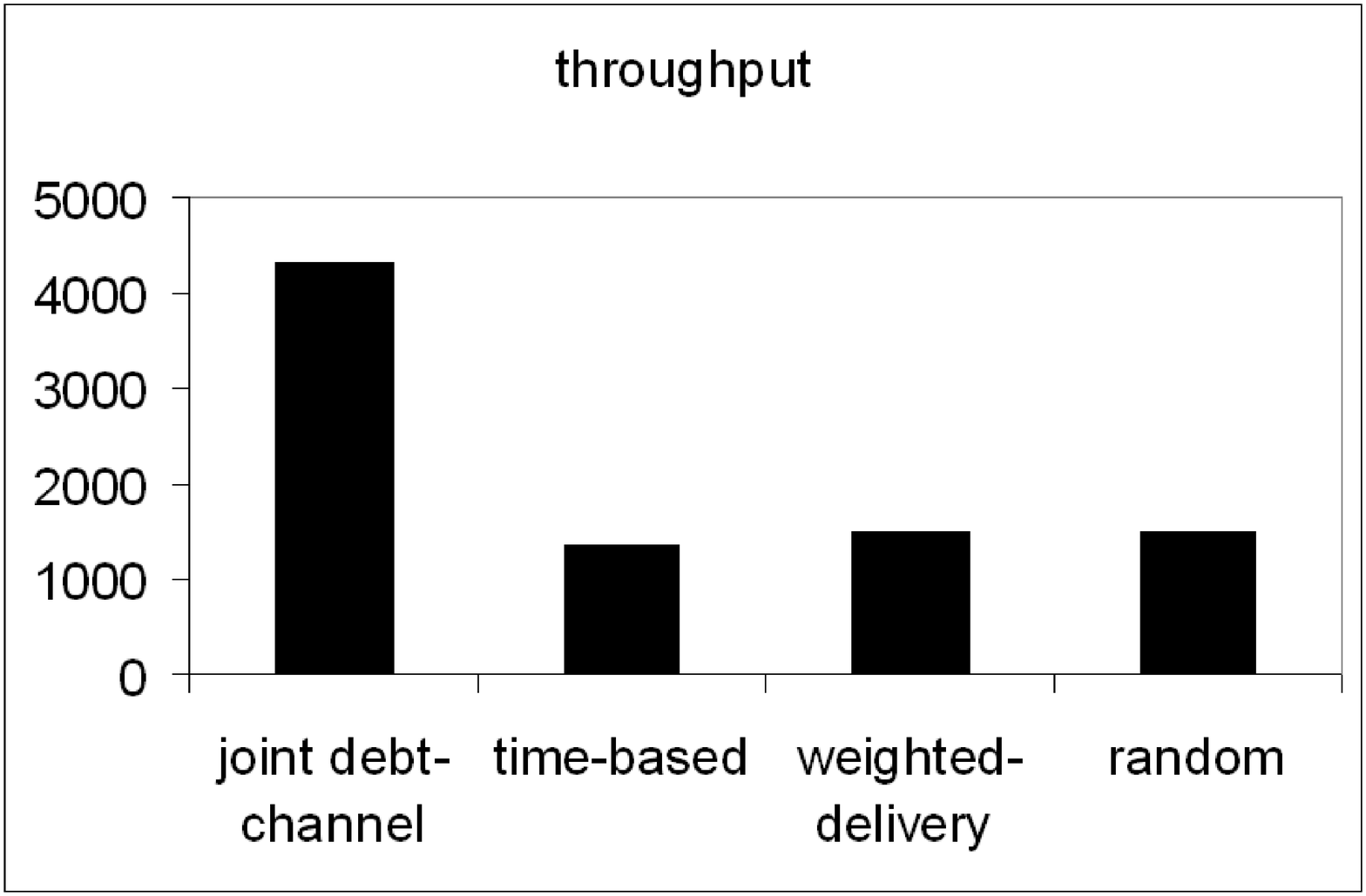}}
\caption{Performance for VoIP traffic under time-varying
channels.}\label{fig:audio_channel}
\end{figure}

For MPEG traffic, we assume there are two groups of clients, with
the same traffic patterns and delivery ratio requirements as those
in the previous section. We use 802.11a with a fixed data rate of 54
Mb/s as the underlying MAC. The mean duration when the channel is in
the bad state is 500 msec for all clients, and the mean duration in
the good state is assumed to be $1+0.5n$ sec for the $n^{th}$ client
in each group. There are 4 clients in both groups.

Simulation results are shown in Figure \ref{fig:video_channel}. As
in the case of VoIP traffic, the joint debt-channel policy incurs
very small total delivery debt while all the other policies have
significantly higher total delivery debts. This result suggests that
the simple modifications of the two largest debt first policies do
not work under time-varying channels. Also, by only scheduling
real-time clients with positive delivery debts, the joint
debt-channel policy achieves higher throughput for the client with
non-real time traffic.

\begin{figure}[t]\subfloat{
\label{fig:video_channel_debt} 
\includegraphics[width=1.6in]{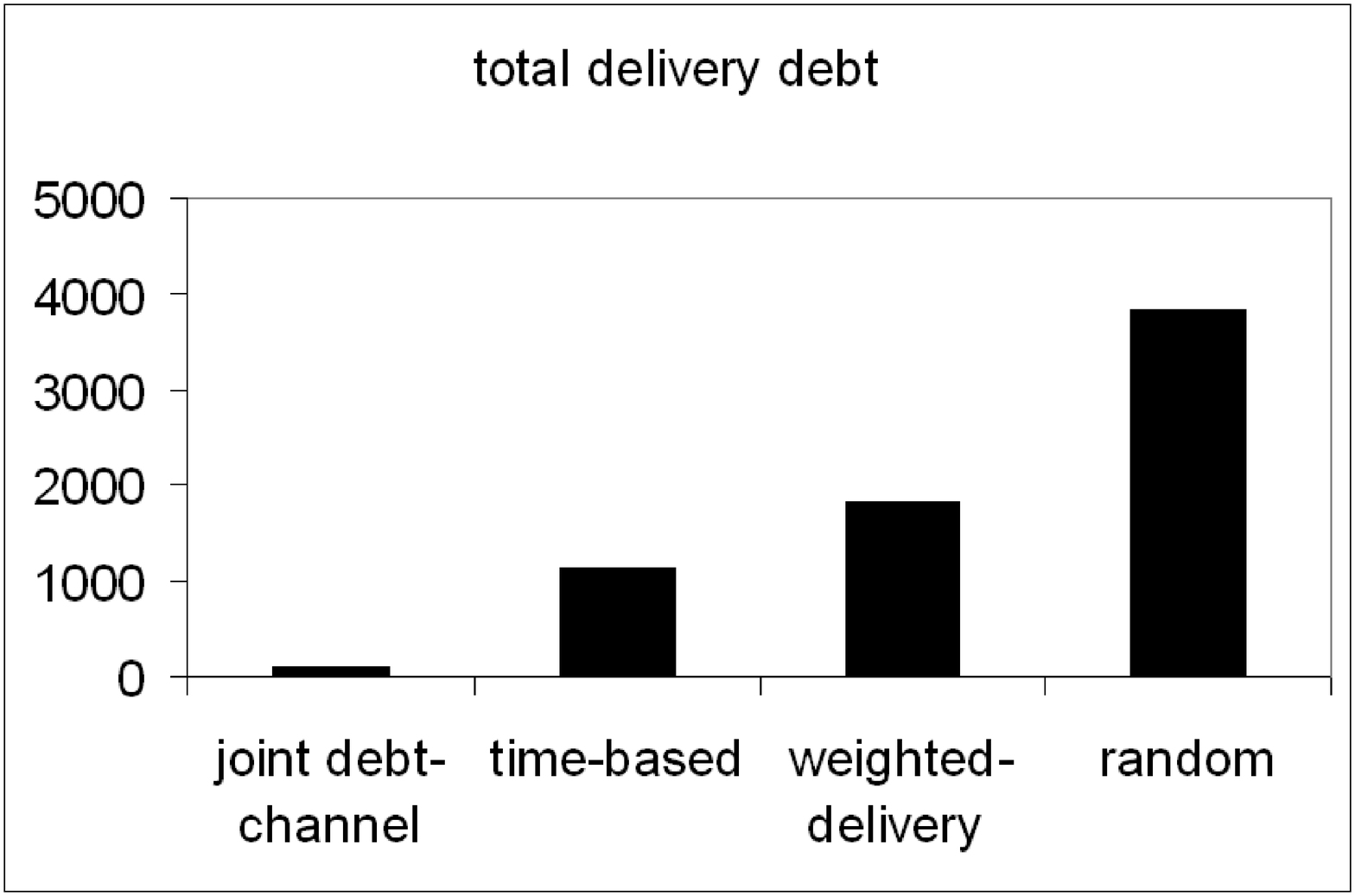}}
\hspace{0.01\linewidth}\subfloat{
\label{fig:video_channel_throughput} 
\includegraphics[width=1.6in]{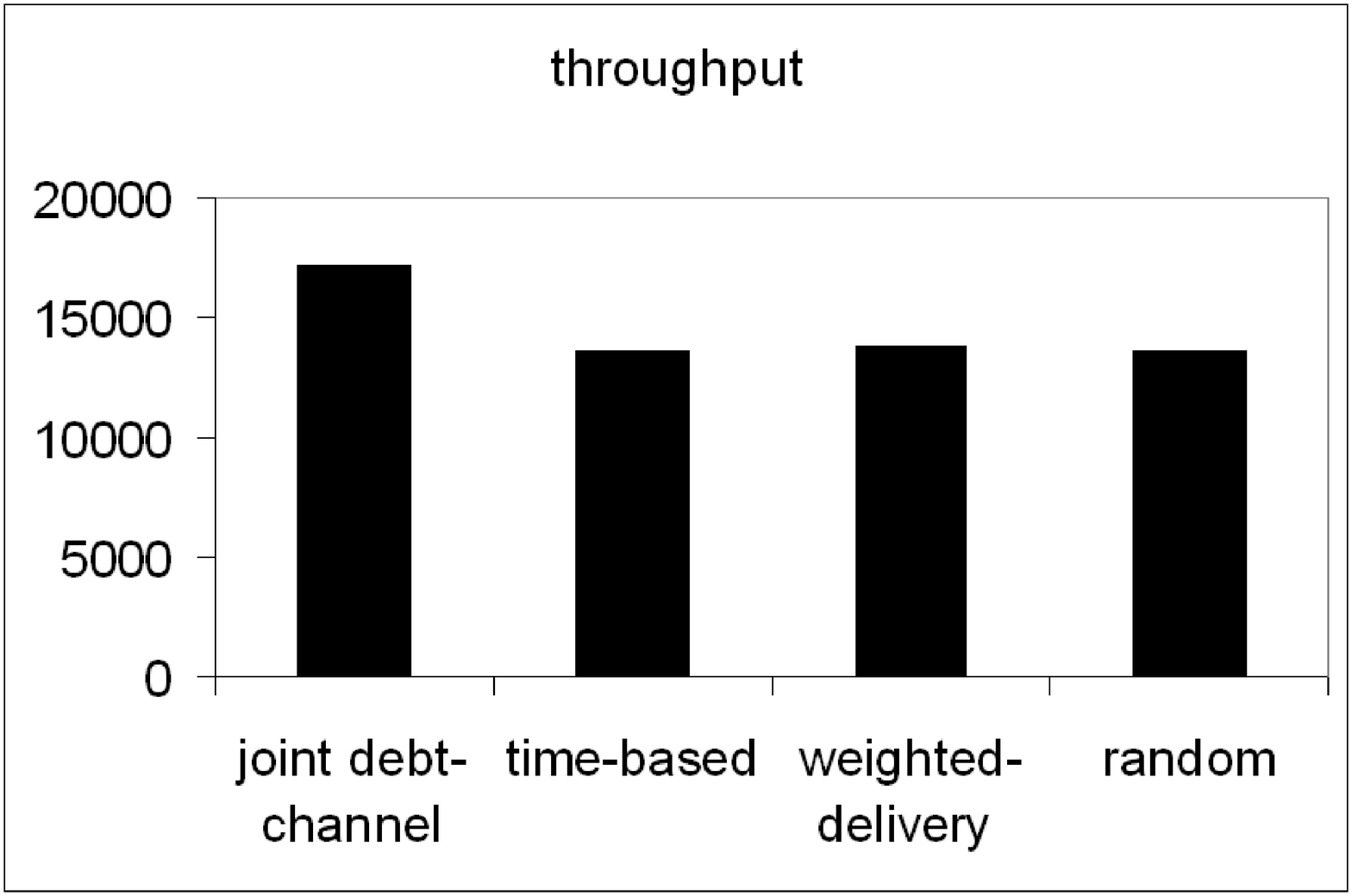}}
\caption{Performance for MPEG traffic under time-varying
channels}\label{fig:video_channel}
\end{figure}

\subsection{Heterogeneous Delay Bounds}
Now, we study the scenario where the channel state is static but
clients require different delay bounds. Since the length of a period
for MPEG traffic is too small, we only simulate VoIP. There are two
groups of clients. All clients generate traffic at rate 64
kbits/sec, and thus each of them has a packet in each period.
Clients in group $A$ require 90$\%$ delivery ratio, with delay
bounds equal to the period length. Clients in group $B$ require
$50\%$ delivery ratio, with delay bounds equal to two-thirds of the
period length, or 22 time slots. The channel reliability for the
$n^{th}$ client in group $A$ is $(84+n)\%$, and that for the
$n^{th}$ client in group $B$ is $(29+n)\%$.

Simulation results are shown in Figure \ref{fig:audio_deadline}. The
adaptive allocation policy has the smallest total delivery debt.
This is because the other policies, especially the two largest debt
first policies, do not consider heterogeneous delay bounds at all.
It is not difficult to see that, to maximize the capacity of the
system, a policy should, in some sense, work in an ``earliest
deadline first'' fashion. Without considering heterogeneous delay
bounds, the largest debt first policies may unwisely schedule
clients with longer delay bounds before those with shorter delay
bounds, and thus result in poor channel utilization. On the other
hand, such poor channel utilization will result in a large number of
idle time slots. Thus, the throughputs for the non-real time traffic
under these policies are higher than those for the adaptive
allocation policy.

\begin{figure}[t]\subfloat{
\label{fig:audio_deadline_debt} 
\includegraphics[width=1.6in]{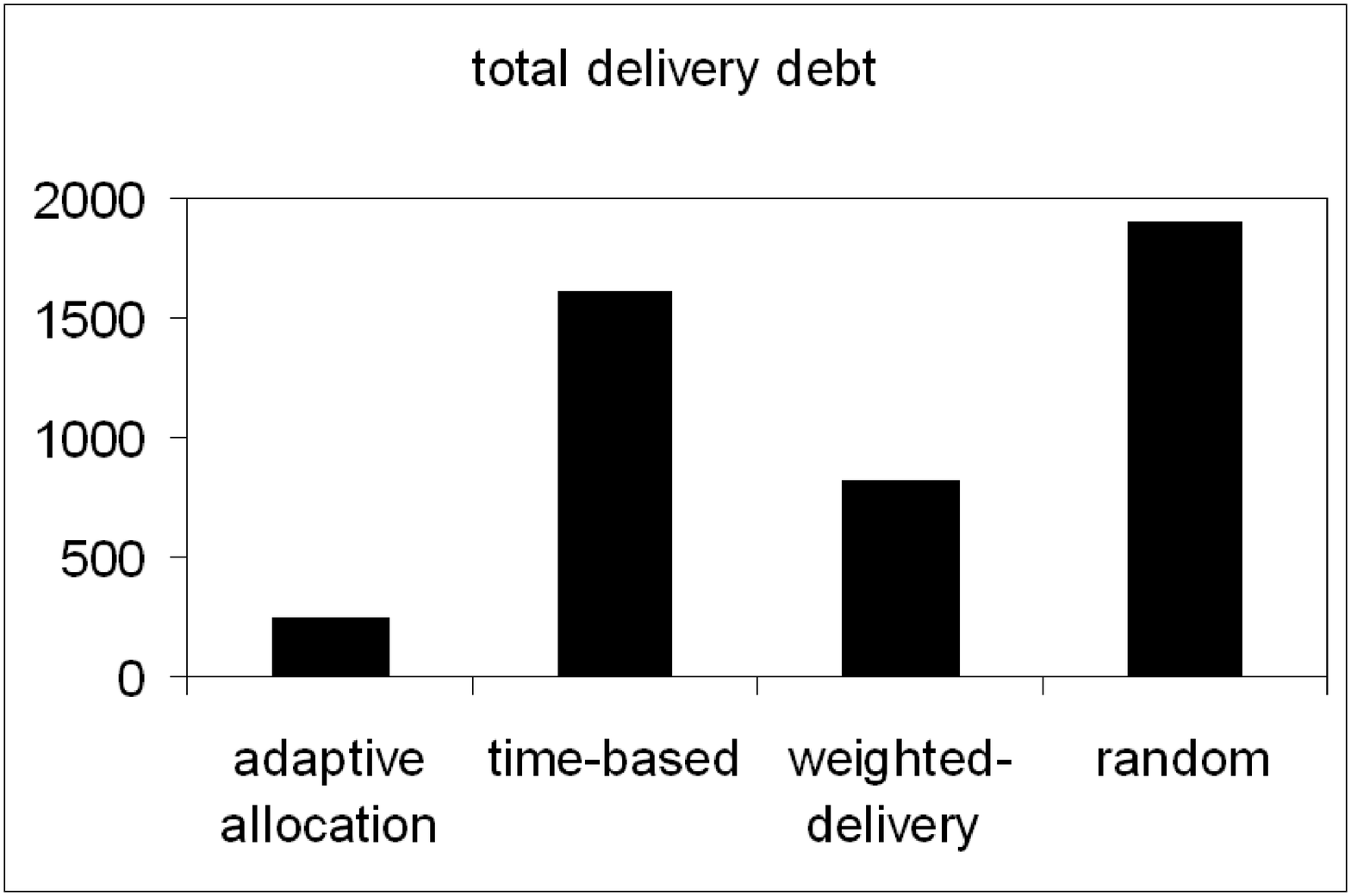}}
\hspace{0.01\linewidth}\subfloat{
\label{fig:audio_deadline_throughput} 
\includegraphics[width=1.6in]{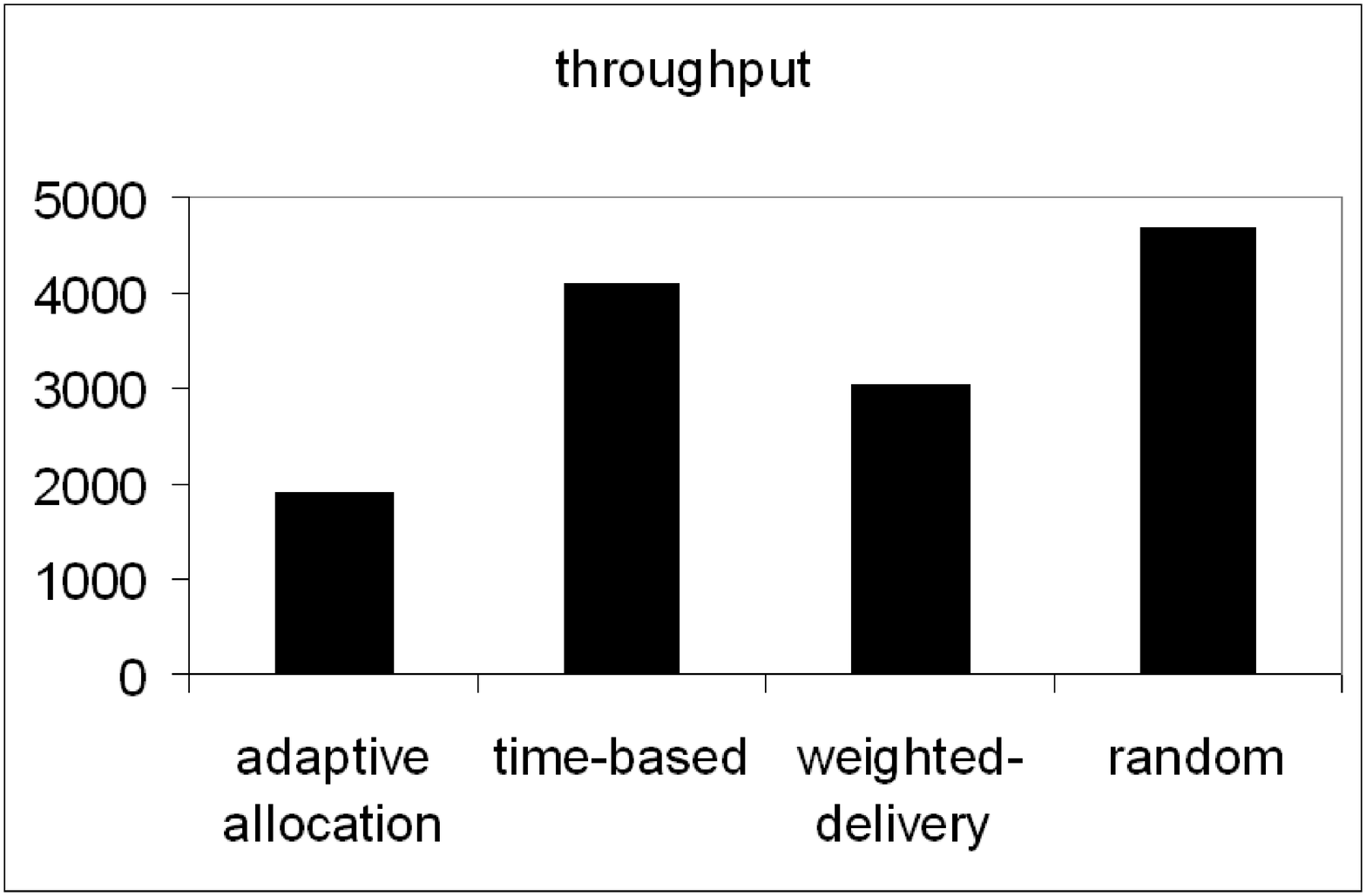}}
\caption{Performance for VoIP traffic under heterogeneous delay
bounds}\label{fig:audio_deadline}
\end{figure}

\section{Conclusion}    \label{section:conclusion}
We have analytically studied the problem of scheduling real-time
traffic over wireless channels. We have extended the model used in
\cite{IHH09MobiHoc} to unreliable wireless channels and real-time
application requirements, including traffic patterns, delay bounds,
and timely-throughput bounds. We have developed a general class of
polices that are feasibility optimal. This class can serve as a
guideline for designing computationally tractable feasibility
optimal policies. We have demonstrated the utility of the class by
deriving scheduling policies for a general case when rate adaptation
is employed and two special cases when it is not, time-varying
channels and heterogeneous delay bounds. Simulation results show
that the policies outperform policies described in
\cite{IHH09MobiHoc}. Thus we have shown not only that the policy
class is useful in designing scheduling policies, but also that
neglecting some realistic and complicated settings can result in
unsatisfactory policies.

\def\baselinestretch{0.85}
\small
\bibliographystyle{plain}
\bibliography{reference}

\def\baselinestretch{1}
\normalsize
\end{document}